\documentclass[11pt,onecolumn]{article}%
\usepackage{amsmath}
\usepackage{amsfonts}
\usepackage{amssymb}
\usepackage{graphicx}
\usepackage{fullpage}%
\providecommand{\U}[1]{\protect\rule{.1in}{.1in}}
\newtheorem{theorem}{Theorem}

\newtheorem{claim}[theorem]{Claim}

\newtheorem{conjecture}[theorem]{Conjecture}

\newtheorem{definition}[theorem]{Definition}

\newtheorem{lemma}[theorem]{Lemma}

\newenvironment{proof}[1][Proof]{\noindent\textbf{#1.} }{\ \rule{0.5em}{0.5em}}
\begin{document}

\title{A Counterexample to the Generalized Linial-Nisan Conjecture}
\author{Scott Aaronson\thanks{MIT. \ Email: aaronson@csail.mit.edu. \ \ This material
is based upon work supported by the National Science Foundation under Grant
No. 0844626. \ Also supported by a DARPA YFA grant and the Keck Foundation.}}
\date{}
\maketitle

\begin{abstract}
In earlier work \cite{aar:ph}, we gave an oracle separating the relational
versions of $\mathsf{BQP}$ and the polynomial hierarchy, and showed that an
oracle separating the decision versions would follow from what we called the
\textit{Generalized Linial-Nisan (GLN) Conjecture}: that \textquotedblleft
almost $k$-wise independent\textquotedblright\ distributions are
indistinguishable from the uniform distribution by constant-depth circuits.
\ The original Linial-Nisan Conjecture was recently proved by Braverman
\cite{braverman}; we offered a \$200 prize for the generalized version. \ In
this paper, we save ourselves \$200 by showing that the GLN Conjecture is
false, at least for circuits of depth $3$ and higher.

As a byproduct, our counterexample also implies that $\mathsf{\Pi}_{2}%
^{p}\not \subset \mathsf{P}^{\mathsf{NP}}$\ relative to a random oracle with
probability $1$. \ It has been conjectured since the 1980s that $\mathsf{PH}%
$\ is infinite relative to a random oracle, but the highest levels of
$\mathsf{PH}$ previously proved separate were $\mathsf{NP}$ and $\mathsf{coNP}%
$.

Finally, our counterexample implies that the famous results of Linial,
Mansour, and Nisan \cite{lmn}, on the structure of $\mathsf{AC}^{0}%
$\ functions, cannot be improved in several interesting respects.

\end{abstract}

\section{Introduction\label{INTRO}}

Proving an oracle separation between $\mathsf{BQP}$ and $\mathsf{PH}$\ is one
of the central open problems of quantum complexity theory. \ In a recent paper
\cite{aar:ph}, we reported the following progress on the problem:

\begin{enumerate}
\item[(1)] We constructed an oracle relative to which $\mathsf{FBQP}%
\not \subset \mathsf{FBPP}^{\mathsf{PH}}$, where $\mathsf{FBQP}$\ and
$\mathsf{FBPP}^{\mathsf{PH}}$\ are the \textquotedblleft
relational\textquotedblright\ versions of $\mathsf{BQP}$\ and $\mathsf{PH}%
$\ respectively (that is, the versions where there are many valid outputs, and
an algorithm's task is to output any one of them).

\item[(2)] We proposed a natural \textit{decision} problem, called
\textsc{Fourier Checking}, which is provably in $\mathsf{BQP}$\ (as an oracle
problem) and which we conjectured was not in $\mathsf{PH}$.

\item[(3)] We showed that \textsc{Fourier Checking}\ has a property called
\textit{almost }$k$\textit{-wise independence}, and that no $\mathsf{BPP}%
_{\mathsf{path}}$\ or $\mathsf{SZK}$ problem shares that property. \ This
allowed us to give oracles relative to which $\mathsf{BQP}$\ was outside those
classes, and to reprove all the known oracle separations between
$\mathsf{BQP}$\ and classical complexity classes in a unified way.

\item[(4)] We conjectured that no $\mathsf{PH}$\ problem has the almost
$k$-wise independence property,\ and called that the \textit{Generalized
Linial-Nisan (GLN) Conjecture}. \ Proving the GLN Conjecture would imply the
existence of an oracle relative to which $\mathsf{BQP}\not \subset
\mathsf{PH}$.
\end{enumerate}

This paper does nothing to modify points (1)-(3) above:\ the unconditional
results in \cite{aar:ph}\ are still true, and we still conjecture not only
that there exists an oracle relative to which $\mathsf{BQP}\not \subset
\mathsf{PH}$, but that \textsc{Fourier Checking}\ is such an oracle.

However, we will show that the hope of proving \textsc{Fourier Checking}%
$\notin\mathsf{PH}$\ by proving the GLN Conjecture was unfounded:

\begin{quotation}
\noindent\textit{The GLN Conjecture is false, at least for }$\mathsf{\Pi}%
_{2}^{p}$\textit{ and higher levels of the polynomial hierarchy}.
\end{quotation}

We prove this by giving an explicit counterexample: a family of depth-three
$\mathsf{AC}^{0}$ circuits that distinguish the uniform distribution over
$n$-bit strings from an $\widetilde{O}\left(  k/n\right)  $-almost $k$-wise
independent\ distribution, with constant bias.\footnote{Note that depth-three
$\mathsf{AC}^{0}$ circuits correspond to the second level of $\mathsf{PH}$,
depth-four circuits correspond to the third level, and so on.}

Our counterexample was inspired by a recent result of Beame and Machmouchi
\cite{beame:quantum}, giving a Boolean function with quantum query complexity
$\Omega\left(  n/\log n\right)  $ that is computable by a depth-three
$\mathsf{AC}^{0}$ circuit. \ This disproved a conjecture, relayed to us
earlier by Beame, stating that every $\mathsf{AC}^{0}$\ function\ has quantum
query complexity $n^{1-\Omega\left(  1\right)  }$. \ Like the Beame-Machmouchi
counterexample, ours involves inputs $X=x_{1}\ldots x_{N}\in\left[  M\right]
^{N}$\ that are lists of positive integers, with the $x_{i}$'s encoded in
binary to obtain a Boolean problem; as well as a function $f:\left[  M\right]
^{N}\rightarrow\left\{  0,1\right\}  $\ that uses two alternating quantifiers
to express a \textquotedblleft global\textquotedblright\ property of $X$. \ In
Beame and Machmouchi's case, the property in question was that the function
$x\left(  i\right)  :=x_{i}$\ is $2$-to-$1$; in our case, the property is that
$x\left(  i\right)  $ is surjective.\footnote{Beame and Machmouchi
\cite{beame:quantum} also mention the surjectivity property, in Corollary 6 of
their paper.}

Our counterexample makes essential use of depth-\textit{three} circuits, and
we find it plausible that the GLN Conjecture still holds for
depth-\textit{two} circuits (i.e., for DNF formulas).\footnote{Indeed, we
originally formulated the conjecture for depth-two circuits only, before
(rashly) extending it to arbitrary depths.} \ As shown in \cite{aar:ph},
proving the GLN Conjecture for depth-two circuits would yield an oracle
relative to which $\mathsf{BQP}\not \subset \mathsf{AM}$, which is already a
longstanding open problem.

Given that the GLN Conjecture resisted attacks for two years (and indirectly
motivated the beautiful works of Razborov \cite{razborov:bazzi}\ and Braverman
\cite{braverman}\ on the original LN Conjecture), our counterexample cannot
have been \textit{quite} as obvious as it seems in retrospect! \ Perhaps Andy
Drucker (personal communication) summarized the situation best: almost
$k$-wise independent distributions seem to be much better at fooling
\textit{people} than at fooling circuits.

\subsection{Further Implications\label{IMP}}

Besides falsifying the GLN Conjecture, our counterexample has several other
interesting implications for $\mathsf{PH}$ and $\mathsf{AC}^{0}$.

Firstly, we are able to use our counterexample to prove that $\left(
\mathsf{\Pi}_{2}^{p}\right)  ^{A}\not \subset \mathsf{P}^{\mathsf{NP}^{A}}%
$\ with probability $1$\ relative to a random oracle $A$. \ Indeed, we
conjecture that our counterexample can even be used to prove $\left(
\mathsf{\Pi}_{2}^{p}\right)  ^{A}\not \subset \left(  \mathsf{\Sigma}_{2}%
^{p}\right)  ^{A}$\ with probability $1$\ for a random oracle $A$. \ The
seminal work of Yao \cite{yao:ph} showed $\mathsf{PH}$\ infinite relative to
\textit{some} oracle,\ but it has been an open problem for almost thirty years
to prove $\mathsf{PH}$\ infinite relative to a \textit{random} oracle (see the
book of H\aa stad \cite{hastad:book} for discussion). \ Motivation for this
problem comes from a surprising result of Book \cite{book}, which\ says that
if $\mathsf{PH}$\ collapses relative to a random oracle, then it also
collapses in the unrelativized world. \ Our result, while simple, appears to
represent the first \textquotedblleft progress\textquotedblright\ toward
separating $\mathsf{PH}$\ by random oracles since the original result of
Bennett and Gill \cite{bg}\ that $\mathsf{P}\neq\mathsf{NP}\neq\mathsf{coNP}%
$\ relative to a random oracle with probability $1$.\footnote{Though
\textquotedblleft working from the opposite direction,\textquotedblright\ Cai
\cite{cai:ph} proved the beautiful result that $\mathsf{PH}\neq\mathsf{PSPACE}%
$\ relative to a random oracle with probability $1$. \ Note that any
relativized world where $\mathsf{PH}$\ is infinite must also satisfy
$\mathsf{PH}\neq\mathsf{PSPACE}$. \ Cai \cite{cai:bh} also proved that
$\mathsf{BH}$\ is infinite with probability $1$, where $\mathsf{BH}%
$\ represents the Boolean hierarchy over $\mathsf{NP}$, a subclass of
$\mathsf{P}_{||}^{\mathsf{NP}}$.}

Secondly, our counterexample shows that the celebrated results of Linial,
Mansour, and Nisan \cite{lmn}, on the Fourier spectrum of $\mathsf{AC}^{0}%
$\ functions, \textit{cannot} be improved in several important respects. \ In
particular, Linial et al.\ showed that every Boolean\ function\ $f:\left\{
0,1\right\}  ^{n}\rightarrow\left\{  0,1\right\}  $\ in $\mathsf{AC}^{0}$\ has
\textit{average sensitivity} $O\left(  \operatorname*{polylog}\left(
n\right)  \right)  $. \ However, we observe that this result fails completely
if we consider a closely-related measure, the\textit{ average
block-sensitivity}. \ Indeed, there exists a reasonably-balanced Boolean
function $f\in\mathsf{AC}^{0}$\ such that every $1$-input can be modified in
$\Omega\left(  n/\log n\right)  $\ disjoint ways to produce a $0$-input, and
almost every $0$-input can be modified in $\Omega\left(  n/\log n\right)  $
disjoint ways to produce a $1$-input. \ What makes this behavior interesting
is that one normally associates it with (say) \textsc{Parity}, the canonical
function \textit{not} in $\mathsf{AC}^{0}$!

Linial et al.\ \cite{lmn}\ also showed that every Boolean function
$f\in\mathsf{AC}^{0}$ has a \textit{low-degree approximating polynomial}: that
is, a real polynomial $p:\left\{  0,1\right\}  ^{n}\rightarrow\mathbb{R}$, of
degree $O\left(  \operatorname*{polylog}\left(  n\right)  \right)  $, such
that%
\[
\operatorname*{E}_{X\in\left\{  0,1\right\}  ^{n}}\left[  \left(  p\left(
X\right)  -f\left(  X\right)  \right)  ^{2}\right]  =o\left(  1\right)  .
\]
However, using our counterexample, we will show that such a polynomial
$p$\ \textit{cannot} generally be written as a linear combination of terms,
$p=\sum_{C}\alpha_{C}C$, where the coefficients satisfy the following bound:%
\[
\sum_{C}\left\vert \alpha_{C}\right\vert 2^{-\left\vert C\right\vert
}=n^{o\left(  1\right)  }.
\]
In other words, such a polynomial cannot be \textquotedblleft
low-fat\textquotedblright\ in the sense defined by Aaronson \cite{aar:ph}%
,\ but must instead involve \textquotedblleft massive
cancellations\textquotedblright\ between positive and negative terms. \ This
gives the first example of a Boolean function $f$\ that can be approximated in
$L_{2}$-norm by a low-degree polynomial, but \textit{not} by a low-degree
low-fat polynomial---thereby answering another one of the open questions from
\cite{aar:ph}.

\subsection{The Future of $\mathsf{BQP}$\ and $\mathsf{PH}$\label{FUTURE}}

While this paper rules out the GLN approach, at least three plausible avenues
remain for proving an oracle separation between $\mathsf{BQP}$ and
$\mathsf{PH}$.

\begin{enumerate}
\item[(1)] Our original idea for proving \textsc{Fourier Checking}%
$\notin\mathsf{PH}$\ was to use a direct random restriction argument---and
while we were unable to make such an argument work, we have also found nothing
to rule it out.

\item[(2)] Besides almost $k$-wise independence, the other \textquotedblleft
obvious\textquotedblright\ property of \textsc{Fourier Checking}\ that might
be useful for lower bounds is its close connection with the \textsc{Majority}%
\ function. \ Indeed, given as input the truth table of a Boolean function
$f:\left\{  0,1\right\}  ^{n}\rightarrow\left\{  -1,1\right\}  $, estimating a
\textit{single} Fourier coefficient $\widehat{f}\left(  s\right)  :=\frac
{1}{2^{n/2}}\sum_{x}\left(  -1\right)  ^{x\cdot s}f\left(  x\right)  $\ is
easily seen to be equivalent to solving \textsc{Majority}, which is known to
be hard for $\mathsf{AC}^{0}$. \ Thus, in proving \textsc{Fourier
Checking}$\notin\mathsf{PH}$, the difficulty is \textquotedblleft
merely\textquotedblright\ to show that checking the answers to $2^{n}$
overlapping \textsc{Majority}\ instances is not significantly easier for an
$\mathsf{AC}^{0}$\ circuit than checking the answer to one instance. \ While
the usual hybrid argument fails in this case, one could hope for some other
reduction---possibly a non-black-box reduction---showing that if
\textsc{Fourier Checking}\ is in $\mathsf{AC}^{0}$, then \textsc{Majority} is
as well.

\item[(3)] Recently, Fefferman and Umans \cite{feffumans} proposed a beautiful
alternative approach to the relativized $\mathsf{BQP}$\ versus $\mathsf{PH}%
$\ question. \ Like approach (2) above, their approach is based on a hoped-for
reduction from \textsc{Majority}. \ However, they replace \textsc{Fourier
Checking}\ by a different candidate problem, which involves Nisan-Wigderson
combinatorial designs \cite{nw} rather than the Fourier transform. \ They show
that their candidate problem is in $\mathsf{BQP}$, and also show that it is
\textit{not} in $\mathsf{PH}$, assuming (roughly speaking) that the analysis
of the NW generator can be improved in a direction that people have wanted to
improve it in for independent reasons. \ Fefferman and Umans' conjecture
follows from the GLN Conjecture,\footnote{As, indeed, \textit{anything}
follows from the GLN Conjecture.} but is much more tailored to a specific
pseudorandom generator, and is completely unaffected by our counterexample.
\end{enumerate}

\subsection{Organization\label{ORG}}

The rest of the paper is organized as follows. \ Section \ref{BACKGROUND}
provides background on $\mathsf{AC}^{0}$, (almost) $k$-wise independence, and
the (Generalized) Linial-Nisan Conjecture; then Section \ref{COUNTER} presents
our counterexample. \ Section \ref{RORACLE} uses the counterexample to prove
that $\mathsf{\Pi}_{2}^{p}\not \subset \mathsf{P}^{\mathsf{NP}}$\ relative to
a random oracle, and Section \ref{POLY} gives implications of the
counterexample for the noise sensitivity and approximate degree of
$\mathsf{AC}^{0}$ functions. \ Section \ref{DISC}\ concludes with some
discussion and open problems.

\section{Background\label{BACKGROUND}}

We refer the reader to \cite{aar:ph} for details on the original and
generalized Linial-Nisan Conjectures, as well as their relationship to
$\mathsf{BQP}$ and $\mathsf{PH}$. \ In this section, we give a brief recap of
the definitions, conjectures, and results that are relevant to our counterexample.

By $\mathsf{AC}^{0}$, we mean the class of Boolean function families $\left\{
f_{n}\right\}  _{n\geq1}$ such that each $f_{n}:\left\{  0,1\right\}
^{n}\rightarrow\left\{  0,1\right\}  $\ is computable by a circuit\ of AND,
OR, and NOT gates with constant depth, unbounded fanin, and size $n^{O\left(
1\right)  }$. \ Here \textit{depth} means the number of alternating layers of
AND and OR gates; NOT gates are not counted. \ Abusing notation, we will often
use phrases like \textquotedblleft$\mathsf{AC}^{0}$ circuit of size
$2^{n^{o\left(  1\right)  }}$,\textquotedblright\ which means the size is now
superpolynomial but the depth is still $O\left(  1\right)  $. \ We will also
generally drop the subscript of $n$.

Throughout the paper we abbreviate probability expressions such as $\Pr
_{X\sim\mathcal{D}}\left[  f\left(  X\right)  \right]  $\ by $\Pr
_{\mathcal{D}}\left[  f\right]  $. \ Let $\mathcal{U}$\ be the uniform
distribution over $n$-bit strings, so that $\Pr_{\mathcal{U}}\left[  X\right]
=1/2^{n}$\ for all $X\in\left\{  0,1\right\}  ^{n}$. \ A
distribution\ $\mathcal{D}$\ over $\left\{  0,1\right\}  ^{n}$ is called
$k$\textit{-wise independent} (for $k\leq n$) if $\mathcal{D}$ is uniform on
every subset of at most $k$\ bits. \ A central question in pseudorandomness
and cryptography is what computational resources are needed to distinguish
such a \textquotedblleft pretend-uniform\textquotedblright\ distribution from
the \textquotedblleft truly-uniform\textquotedblright\ one. \ In 1990, Linial
and Nisan \cite{ln}\ famously conjectured that $n^{\varepsilon}$\textit{-wise
independence fools }$\mathsf{AC}^{0}$\textit{\ circuits:}

\begin{conjecture}
[Linial-Nisan or LN Conjecture]\label{ln}Let $\mathcal{D}$\ be any
$n^{\Omega\left(  1\right)  }$-wise independent distribution over $\left\{
0,1\right\}  ^{n}$, and let $f:\left\{  0,1\right\}  ^{n}\rightarrow\left\{
0,1\right\}  $\ be computed by an $\mathsf{AC}^{0}$ circuit of size
$2^{n^{o\left(  1\right)  }}$. \ Then%
\[
\left\vert \Pr_{\mathcal{D}}\left[  f\right]  -\Pr_{\mathcal{U}}\left[
f\right]  \right\vert =o\left(  1\right)  .
\]

\end{conjecture}

(The actual parameters in the LN Conjecture are considerably stronger than the
above, but also more complicated to state. \ We chose weaker parameters that
suffice for our discussion.)

After seventeen years of almost no progress, Bazzi \cite{bazzi} finally proved
Conjecture \ref{ln}\ for the special case of depth-$2$ circuits. \ Shortly
afterward, Razborov \cite{razborov:bazzi} gave a dramatically simpler proof of
Bazzi's theorem, and shortly after \textit{that}, Braverman \cite{braverman}
proved the full Conjecture \ref{ln}:

\begin{theorem}
[Braverman's Theorem \cite{braverman}]\label{bravermanthm}Let $f:\left\{
0,1\right\}  ^{n}\rightarrow\left\{  0,1\right\}  $\ be computed by an
$\mathsf{AC}^{0}$ circuit of size $S$\ and depth $d$, and let $\mathcal{D}%
$\ be a $\left(  \log\frac{S}{\varepsilon}\right)  ^{7d^{2}}$-wise independent
distribution over $\left\{  0,1\right\}  ^{n}$.\ \ Then for all sufficiently
large $S$,%
\[
\left\vert \Pr_{\mathcal{D}}\left[  f\right]  -\Pr_{\mathcal{U}}\left[
f\right]  \right\vert \leq\varepsilon.
\]

\end{theorem}

Even before the work of Razborov \cite{razborov:bazzi}\ and Braverman
\cite{braverman}, we had proposed a deceptively modest-seeming generalization
of Conjecture \ref{ln}, motivated by the application to the $\mathsf{BQP}$
versus $\mathsf{PH}$\ question mentioned previously. \ To state the
generalization, we need some more terminology. \ Let $X=x_{1}\ldots x_{n}%
\in\left\{  0,1\right\}  ^{n}$ be a string. \ Then a \textit{literal} is an
expression of the form $x_{i}$ or $1-x_{i}$, and a $k$-\textit{term} is a
product of $k$ literals (each involving a different $x_{i}$), which is $1$ if
the literals all take on prescribed values and $0$ otherwise.

\begin{definition}
[almost $k$-wise independence]Given a distribution $\mathcal{D}$ over
$\left\{  0,1\right\}  ^{n}$ and a $k$-term $C$, we say that $C$ is
$\varepsilon$-fooled\ by $\mathcal{D}$\ if%
\[
1-\varepsilon\leq\frac{\Pr_{\mathcal{D}}\left[  C\right]  }{\Pr_{\mathcal{U}%
}\left[  C\right]  }\leq1+\varepsilon.
\]
(Note that $\Pr_{\mathcal{U}}\left[  C\right]  $\ is just $2^{-k}$.) \ Then
$\mathcal{D}$\ is $\varepsilon$\textit{-almost }$k$\textit{-wise independent}
if every $k$-term $C$ is $\varepsilon$-fooled\ by $\mathcal{D}$.
\end{definition}

In other words, there should be no assignment to any $k$ bits, such that
conditioning on that assignment gives us much information about whether $X$
was drawn from $\mathcal{D}$\ or from $\mathcal{U}$. \ We can now state the
conjecture that we falsify.

\begin{conjecture}
[Generalized Linial-Nisan or GLN Conjecture]\label{gln}Let $\mathcal{D}$\ be a
$1/n^{\Omega\left(  1\right)  }$-almost $n^{\Omega\left(  1\right)  }$-wise
independent distribution over $\left\{  0,1\right\}  ^{n}$, and let
$f:\left\{  0,1\right\}  ^{n}\rightarrow\left\{  0,1\right\}  $\ be computed
by an $\mathsf{AC}^{0}$ circuit of size $2^{n^{o\left(  1\right)  }}$. \ Then%
\[
\left\vert \Pr_{\mathcal{D}}\left[  f\right]  -\Pr_{\mathcal{U}}\left[
f\right]  \right\vert =o\left(  1\right)  .
\]

\end{conjecture}

Note that, for Conjecture \ref{gln}\ not to be ruled out \textit{immediately},
it is essential that our definition of $\varepsilon$-fooling was
\textit{multiplicative} rather than additive. \ For suppose we had merely
required that, on every subset of indices $S\subseteq\left[  n\right]  $\ with
$\left\vert S\right\vert \leq k$, the marginal distribution $\mathcal{D}_{S}%
$\ was $\varepsilon$-close in variation distance to the uniform distribution.
\ Then it would be easy to construct almost $k$-wise distributions
$\mathcal{D}$\ that were distinguishable from the uniform distribution even by
DNF formulas. \ For example, the uniform distribution over all
sequences\ $X=x_{1}\ldots x_{N}\in\left[  N\right]  ^{N}$\ that are
\textit{permutations} (with the $x_{i}$'s appropriately coded in binary) is
one such $\mathcal{D}$.

This paper shows that, even with the more careful multiplicative definition of
$\varepsilon$-fooling, there is \textit{still} a counterexample to Conjecture
\ref{gln}---although we have to work harder and use higher-depth circuits to
construct it. \ The failure of Conjecture \ref{gln}\ means that Braverman's
Theorem\ is \textquotedblleft essentially optimal,\textquotedblright\ in the
sense that one cannot relax the $k$-wise independence condition to almost
$k$-wise independence. \ This demonstrates a striking contrast between
$k$-wise independence and almost $k$-wise independence in terms of their
implications for pseudorandomness.

\section{The Counterexample\label{COUNTER}}

Fix a positive integer $m$, and let $M:=2^{m}$. \ Then it will be useful think
of the input $X=x_{1}\ldots x_{N}$\ as belonging to the set $\left[  M\right]
^{N}$, where $N:=\left\lceil Mm\ln2\right\rceil $. \ However, to make contact
with the original statement of the GLN Conjecture, we can easily encode such
an $X$\ as an $n$-bit string where $n:=Nm$, by writing out each $x_{i}$\ in
binary. \ Abusing notation, we will speak interchangeably about $X$\ as an
element of$\ \left\{  0,1\right\}  ^{n}$ or of $\left[  M\right]  ^{N}$.

Let the \textit{image} of $X$, or $\operatorname{Im}_{X}:=\left\{
x_{1},\ldots,x_{N}\right\}  $, be the set of integers that appear in $X$.
\ Then define the \textit{surjectivity function}, $f_{\operatorname*{Surj}%
}:\left\{  0,1\right\}  ^{n}\rightarrow\left\{  0,1\right\}  $\ by
$f_{\operatorname*{Surj}}\left(  X\right)  =1$\ if $\operatorname{Im}%
_{X}=\left[  M\right]  $\ and $f_{\operatorname*{Surj}}\left(  X\right)
=0$\ otherwise. \ A first easy observation is that $f_{\operatorname*{Surj}%
}\in\mathsf{AC}^{0}$.

\begin{lemma}
\label{inac0}$f_{\operatorname*{Surj}}$ is computable by an $\mathsf{AC}^{0}%
$\ circuit of depth $3$ and size $O\left(  NMm\right)  $.
\end{lemma}

\begin{proof}
For all $i\in\left[  N\right]  $\ and $y\in\left[  M\right]  $, let
$\Delta\left(  x_{i},y\right)  $\ denote the $m$-term that evaluates to $1$ if
$x_{i}=y$\ and to $0$ otherwise. \ Then%
\[
f_{\operatorname*{Surj}}\left(  X\right)  =\bigwedge_{y\in\left[  M\right]
}\bigvee_{i\in\left[  N\right]  }\Delta\left(  x_{i},y\right)  .
\]

\end{proof}

Now let $\mathcal{U}$\ be the uniform distribution over $\left[  M\right]
^{N}$, so that $\Pr_{\mathcal{U}}\left[  X\right]  =1/M^{N}$\ for all
$X\in\left[  M\right]  ^{N}$. \ Also, given an input $X\in\left[  M\right]
^{N}$, we define a distribution $\mathcal{D}\left(  X\right)  $\ over
\textquotedblleft perturbed\textquotedblright\ versions of $X$ via the
following procedure:

\begin{enumerate}
\item[(1)] Choose $y$\ uniformly at random from $\left[  M\right]  $.

\item[(2)] For each $i\in\left[  N\right]  $\ such that $x_{i}=y$, change
$x_{i}$\ to a uniform, independent sample from $\left[  M\right]
\diagdown\left\{  y\right\}  $.
\end{enumerate}

Then we let $\mathcal{D}:=\mathcal{D}\left(  \mathcal{U}\right)  $\ be the
distribution over inputs $Z$\ obtained by first drawing an $X$\ from
$\mathcal{U}$, and then sampling $Z$\ from $\mathcal{D}\left(  X\right)  $.
\ Notice that $\operatorname{Im}_{Z}\neq\left[  M\right]  $\ and hence
$f\left(  Z\right)  =0$\ for all $Z$ in the support of $\mathcal{D}$.

Here is an observation that will be helpful later. \ Given a sample
$Z=z_{1}\ldots z_{N}$\ from $\mathcal{D}$,\ we can define a distribution
$\mathcal{D}^{\operatorname*{inv}}\left(  Z\right)  $\ over
perturbed\ versions of $Z$ via the following \textquotedblleft
inverse\textquotedblright\ procedure:

\begin{enumerate}
\item[(1)] Choose $y$\ uniformly at random from $\left[  M\right]
\diagdown\operatorname{Im}_{Z}$.

\item[(2)] For each $i\in\left[  N\right]  $, change $z_{i}$\ to $y$ with
independent probability $1/M$.
\end{enumerate}

We claim that $\mathcal{D}^{\operatorname*{inv}}$\ is indeed the inverse of
$\mathcal{D}$.

\begin{claim}
\label{invclaim}$\mathcal{D}^{\operatorname*{inv}}\left(  \mathcal{D}\left(
\mathcal{U}\right)  \right)  =\mathcal{U}.$
\end{claim}

\begin{proof}
Let $\mathcal{D}_{y}\left(  X\right)  $\ be the variant of $\mathcal{D}\left(
X\right)  $\ where we fix the element $y\in\left[  M\right]  $\ in step (1),
so that $\mathcal{D}\left(  X\right)  =\operatorname*{E}_{y\in\left[
M\right]  }\mathcal{D}_{y}\left(  X\right)  $. \ Similarly, let $\mathcal{D}%
_{y}^{\operatorname*{inv}}\left(  Z\right)  $\ be the variant of
$\mathcal{D}^{\operatorname*{inv}}\left(  Z\right)  $\ where we fix the
element $y\in\left[  M\right]  \diagdown\operatorname{Im}_{Z}$. \ Then it is
easy to see that, for every fixed $y\in\left[  M\right]  $, we have
$\mathcal{D}_{y}^{\operatorname*{inv}}\left(  \mathcal{D}_{y}\left(
\mathcal{U}\right)  \right)  =\mathcal{U}$. \ For choosing each $x_{i}%
$\ uniformly at random, then changing it randomly if equals $y$, then changing
it \textit{back} to $y$\ with probability $1/M$, is just a more complicated
way of choosing $x_{i}$\ uniformly at random.

Now let $\operatorname*{Hist}\left(  X\right)  $ be the \textit{histogram} of
$X$: that is, the multiset $\left\{  h_{1},\ldots,h_{M}\right\}  $\ where
$h_{y}:=\left\vert \left\{  i:x_{i}=y\right\}  \right\vert $. \ Then we can
conclude from the above that, for every $y\in\left[  M\right]  $,%
\begin{align*}
\operatorname*{Hist}\left(  \mathcal{D}^{\operatorname*{inv}}\left(
\mathcal{D}\left(  \mathcal{U}\right)  \right)  \right)   &
=\operatorname*{Hist}\left(  \mathcal{D}^{\operatorname*{inv}}\left(
\mathcal{D}_{y}\left(  \mathcal{U}\right)  \right)  \right)  \\
&  =\operatorname*{Hist}\left(  \mathcal{D}_{y}^{\operatorname*{inv}}\left(
\mathcal{D}_{y}\left(  \mathcal{U}\right)  \right)  \right)  \\
&  =\operatorname*{Hist}\left(  \mathcal{U}\right)  .
\end{align*}
Call a distribution $\mathcal{A}$\ over $\left[  M\right]  ^{N}$%
\ \textit{symmetric} if $\Pr_{\mathcal{A}}\left[  X\right]  $\ depends only on
$\operatorname*{Hist}\left(  X\right)  $. \ Notice that $\mathcal{U}$\ is
symmetric, and that if $\mathcal{A}$\ is symmetric, then $\mathcal{D}\left(
\mathcal{A}\right)  $\ and $\mathcal{D}^{\operatorname*{inv}}\left(
\mathcal{A}\right)  $\ are both symmetric also. \ This means that from
$\operatorname*{Hist}\left(  \mathcal{D}^{\operatorname*{inv}}\left(
\mathcal{D}\left(  \mathcal{U}\right)  \right)  \right)  =\operatorname*{Hist}%
\left(  \mathcal{U}\right)  $, we can conclude that $\mathcal{D}%
^{\operatorname*{inv}}\left(  \mathcal{D}\left(  \mathcal{U}\right)  \right)
=\mathcal{U}$ as well.
\end{proof}

We now show that the function $f_{\operatorname*{Surj}}$ distinguishes
$\mathcal{D}$\ from $\mathcal{U}$ with constant bias.

\begin{lemma}
\label{dist}$\operatorname*{E}_{\mathcal{U}}\left[  f_{\operatorname*{Surj}%
}\right]  -\operatorname*{E}_{\mathcal{D}}\left[  f_{\operatorname*{Surj}%
}\right]  \geq1/e-o\left(  1\right)  .$
\end{lemma}

\begin{proof}
By construction, we have $\operatorname*{E}_{\mathcal{D}}\left[
f_{\operatorname*{Surj}}\right]  =0$. \ On the other hand,%
\[
\operatorname*{E}_{\mathcal{U}}\left[  f_{\operatorname*{Surj}}\right]
=\Pr_{\mathcal{U}}\left[  \left\vert \operatorname{Im}_{X}\right\vert
=M\right]  .
\]
Think of $N=M\ln M+O\left(  1\right)  $ balls, which are thrown uniformly and
independently into $M$ bins.\ \ Then $\left\vert \operatorname{Im}%
_{X}\right\vert $\ is just the number of bins that receive at least one ball.
\ Using the Poisson approximation, we have%
\[
\lim_{M\rightarrow\infty}\Pr_{\mathcal{U}}\left[  \left\vert \operatorname{Im}%
_{X}\right\vert =M\right]  =\frac{1}{e},
\]
and therefore $\operatorname*{E}_{\mathcal{U}}\left[  f_{\operatorname*{Surj}%
}\right]  \geq1/e-o\left(  1\right)  $.
\end{proof}

To show that the distribution $\mathcal{D}$\ is almost $k$-wise independent,
we first need a technical claim, to the effect that almost $k$-wise
independence behaves well with\ respect to restrictions. \ Given a $k$-term
$C$, let $V\left(  C\right)  $\ be the set of variables that occur in $C$.
\ Also, given a set $S$ of variables that contains $V\left(  C\right)  $, let
$U_{S}\left(  C\right)  $\ be the set of all $2^{\left\vert S\right\vert -k}$
terms $B$\ such that $V\left(  B\right)  =S$ and $B\Longrightarrow C$.

\begin{claim}
\label{induc}Given a $k$-term $C$ and a set $S$ containing $V\left(  C\right)
$, suppose every term $B\in U_{S}\left(  C\right)  $\ is $\varepsilon$-fooled
by $\mathcal{D}$. \ Then $C$ is $\varepsilon$-fooled\ by $\mathcal{D}$.
\end{claim}

\begin{proof}
It suffices to check the claim in the case $\left\vert S\right\vert =k+1$,
since we can then use induction on $k$. \ Let $S=V\left(  C\right)
\cup\left\{  x\right\}  $ for some variable $x\notin V\left(  C\right)  $.
\ Then $U_{S}\left(  C\right)  $\ contains two terms: $C_{0}:=C\wedge
\overline{x}$\ and $C_{1}:=C\wedge x$. \ By the law of total probability, we
have $\Pr_{\mathcal{D}}\left[  C\right]  =\Pr_{\mathcal{D}}\left[
C_{0}\right]  +\Pr_{\mathcal{D}}\left[  C_{1}\right]  $\ and $\Pr
_{\mathcal{U}}\left[  C\right]  =\Pr_{\mathcal{U}}\left[  C_{0}\right]
+\Pr_{\mathcal{U}}\left[  C_{1}\right]  $. \ Hence%
\[
\min\left\{  \frac{\Pr_{\mathcal{D}}\left[  C_{0}\right]  }{\Pr_{\mathcal{U}%
}\left[  C_{0}\right]  },\frac{\Pr_{\mathcal{D}}\left[  C_{1}\right]  }%
{\Pr_{\mathcal{U}}\left[  C_{1}\right]  }\right\}  \leq\frac{\Pr_{\mathcal{D}%
}\left[  C\right]  }{\Pr_{\mathcal{U}}\left[  C\right]  }\leq\max\left\{
\frac{\Pr_{\mathcal{D}}\left[  C_{0}\right]  }{\Pr_{\mathcal{U}}\left[
C_{0}\right]  },\frac{\Pr_{\mathcal{D}}\left[  C_{1}\right]  }{\Pr
_{\mathcal{U}}\left[  C_{1}\right]  }\right\}  .
\]
So if $C_{0}$\ and $C_{1}$\ are both $\varepsilon$-fooled by $\mathcal{D}$,
then $C$\ is $\varepsilon$-fooled as well.
\end{proof}

Given an input $X=x_{1}\ldots x_{N}$, recall that $\Delta\left(
x_{i},y\right)  $\ denotes a term that evaluates to $1$ if $x_{i}=y$, and to
$0$ if $x_{i}\neq y$. \ Then let a \textit{proper }$k$\textit{-term}\ $C$\ be
a product of the form $\Delta\left(  x_{i_{1}},y_{1}\right)  \cdot\cdots
\cdot\Delta\left(  x_{i_{k}},y_{k}\right)  $, where $1\leq i_{1}<\cdots
<i_{k}\leq N$\ and $y_{1},\ldots,y_{k}\in\left[  M\right]  $.

We now prove the central fact, that $\mathcal{D}$\ is almost $k$-wise independent.

\begin{lemma}
\label{indep}$\mathcal{D}$\ is $2k/M$-almost $k$-wise independent for all
$k\leq M/2$.
\end{lemma}

\begin{proof}
Notice that a Boolean $k$-term can involve bits from at most $k$ different
$x_{i}$'s. \ So by Claim \ref{induc}, to show that any Boolean $k$-term is
$\varepsilon$-fooled by $\mathcal{D}$, it suffices to show that any
\textit{proper} $k$-term%
\[
C=\Delta\left(  x_{i_{1}},y_{1}\right)  \cdot\cdots\cdot\Delta\left(
x_{i_{k}},y_{k}\right)
\]
is $\varepsilon$-fooled by $\mathcal{D}$.

We first upper-bound $\Pr_{\mathcal{D}}\left[  C\right]  $. \ Recall that to
sample an input $Z$\ from the distribution $\mathcal{D}$, we first sample an
$X$\ from $\mathcal{U}$, and then sample $Z$\ from $\mathcal{D}\left(
X\right)  $. \ Suppose $C\left(  X\right)  =1$. \ Then the only way we can get
$C\left(  Z\right)  =0$\ is if, when we perturb the input $X$ to obtain
$\mathcal{D}\left(  X\right)  $, some $\Delta\left(  x_{i_{j}},y_{j}\right)
$\ changes from TRUE to FALSE. \ But for each $j\in\left[  k\right]  $, this
can happen only if $y=y_{j}$, which occurs with probability $1/M$. \ So by the
union bound,%
\[
\Pr_{\mathcal{D}}\left[  C\right]  \geq\Pr_{\mathcal{U}}\left[  C\right]
\cdot\left(  1-\frac{k}{M}\right)  .
\]

We can similarly upper-bound $\Pr_{\mathcal{U}}\left[  C\right]  $. \ By Claim
\ref{invclaim}, to sample an input $X$\ from $\mathcal{U}$, we can first
sample a $Z$\ from $\mathcal{D}$, and then sample $X$ from $\mathcal{D}%
^{\operatorname*{inv}}\left(  Z\right)  $. \ Suppose $C\left(  Z\right)  =1$.
\ Then we can only get $C\left(  X\right)  =0$\ if, when we perturb $Z$ to
$\mathcal{D}^{\operatorname*{inv}}\left(  Z\right)  $, some $\Delta\left(
z_{i_{j}},y_{j}\right)  $\ changes from TRUE to FALSE. \ But each $z_{i}%
$\ changes with probability at most $1/M$. \ So by the union bound,%
\[
\Pr_{\mathcal{U}}\left[  C\right]  \geq\Pr_{\mathcal{D}}\left[  C\right]
\cdot\left(  1-\frac{k}{M}\right)  .
\]

Combining the upper and lower bounds, and using the fact that $k\leq M/2$, we
have%
\[
1-\frac{k}{M}\leq\frac{\Pr_{\mathcal{D}}\left[  C\right]  }{\Pr_{\mathcal{U}%
}\left[  C\right]  }\leq1+\frac{2k}{M}.
\]

\end{proof}

Combining Lemmas \ref{inac0}, \ref{dist}, and \ref{indep}, and recalling that
$n=Nm$, we obtain the following.

\begin{theorem}
\label{glnfalse}Conjecture \ref{gln}\ (the GLN Conjecture) is false. \ Indeed,
there exists a family of Boolean functions $f_{\operatorname*{Surj}}:\left\{
0,1\right\}  ^{n}\rightarrow\left\{  0,1\right\}  $, computable by
$\mathsf{AC}^{0}$ circuits\ of size $O\left(  n^{2}\right)  $, depth $3$, and
bottom fanin $O\left(  \log n\right)  $, as well as an $O\left(  \left(
k\log^{2}n\right)  /n\right)  $-almost $k$-wise independent distribution
$\mathcal{D}$\ over $\left\{  0,1\right\}  ^{n}$, such that $\operatorname*{E}%
_{\mathcal{D}}\left[  f_{\operatorname*{Surj}}\right]  -\operatorname*{E}%
_{\mathcal{U}}\left[  f_{\operatorname*{Surj}}\right]  =\Omega\left(
1\right)  $.
\end{theorem}

\section{Random Oracle Separations\label{RORACLE}}

In this section, we reuse the function $f_{\operatorname*{Surj}}$\ and
distribution $\mathcal{D}$\ from Section \ref{COUNTER}\ to show that $\left(
\mathsf{\Pi}_{2}^{p}\right)  ^{A}\not \subset \mathsf{P}^{\mathsf{NP}^{A}}%
$\ with probability $1$\ relative to a random oracle $A$. \ The central
observation here is simply that $\mathcal{D}$\ has support on a
\textit{constant} fraction of $\left[  M\right]  ^{N}$---and that therefore,
any algorithm that computes $f_{\operatorname*{Surj}}\left(  X\right)  $ on a
$1-\varepsilon$ fraction of inputs $X\in\left[  M\right]  ^{N}$\ must also
distinguish $\mathcal{D}$\ from $\mathcal{U}$\ with constant bias. \ The
following lemma makes this implication precise.

\begin{lemma}
\label{dvsu}Let $B$ be a random variable such that $\Pr_{\mathcal{U}}\left[
B=f_{\operatorname*{Surj}}\right]  \geq0.92$. \ Then $\Pr_{\mathcal{U}}\left[
B\right]  -\Pr_{\mathcal{D}}\left[  B\right]  \geq0.022-o\left(  1\right)  $.
\end{lemma}

\begin{proof}
For convenience, let us adopt the convention that all probabilities are
implicitly the limiting probabilities as $m\rightarrow\infty$; this introduces
at most an $o\left(  1\right)  $ additive error. \ Then $\Pr_{\mathcal{U}%
}\left[  f_{\operatorname*{Surj}}\right]  =1/e$, so%
\[
\Pr_{\mathcal{U}}\left[  B\right]  \geq\Pr_{\mathcal{U}}\left[
f_{\operatorname*{Surj}}\right]  -\Pr_{\mathcal{U}}\left[  B\neq
f_{\operatorname*{Surj}}\right]  \geq\frac{1}{e}-0.08>0.287.
\]
It remains to upper-bound $\Pr_{\mathcal{D}}\left[  B\right]  $. \ Using the
Poisson approximation, for every fixed integer $k\geq0$\ we have%
\[
\Pr_{\mathcal{U}}\left[  \left\vert \operatorname{Im}_{X}\right\vert
=M-k\right]  =\frac{1}{e\cdot k!}.
\]
By comparison, for every fixed $k\geq1$\ we have%
\[
\Pr_{\mathcal{D}}\left[  \left\vert \operatorname{Im}_{X}\right\vert
=M-k\right]  =\frac{1}{e\cdot\left(  k-1\right)  !}.
\]
Now, once we condition on the value of $\left\vert \operatorname{Im}%
_{X}\right\vert $, it is not hard to see that the distributions $\mathcal{D}%
$\ and $\mathcal{U}$\ are identical. \ Thus, since%
\[
\frac{\Pr_{\mathcal{D}}\left[  \left\vert \operatorname{Im}_{X}\right\vert
=M-k\right]  }{\Pr_{\mathcal{U}}\left[  \left\vert \operatorname{Im}%
_{X}\right\vert =M-k\right]  }=\frac{e\cdot k!}{e\cdot\left(  k-1\right)  !}=k
\]
increases with $k$, the way to maximize $\Pr_{\mathcal{D}}\left[  B\right]  $
is to set $B=1$\ for those inputs $X$ such that $k$\ is as large as possible
(in other words, such that $\left\vert \operatorname{Im}_{X}\right\vert $\ is
as small as possible). \ Notice that%
\begin{align*}
\Pr_{\mathcal{U}}\left[  \left(  \left\vert \operatorname{Im}_{X}\right\vert
<M\right)  \wedge B\right]    & \leq\Pr_{\mathcal{U}}\left[  B\neq
f_{\operatorname*{Surj}}\right]  \\
& \leq0.08\\
& <1-\frac{5}{2e}\\
& =\sum_{k=3}^{\infty}\frac{1}{e\cdot k!}.
\end{align*}
It follows that%
\begin{align*}
\Pr_{\mathcal{D}}\left[  B\right]    & \leq\sum_{k=3}^{\infty}\Pr
_{\mathcal{D}}\left[  \left\vert \operatorname{Im}_{X}\right\vert =M-k\right]
\\
& =\sum_{k=3}^{\infty}\frac{1}{e\cdot\left(  k-1\right)  !}\\
& =1-\frac{2}{e}\\
& <0.265.
\end{align*}
Combining,%
\[
\Pr_{\mathcal{U}}\left[  B\right]  -\Pr_{\mathcal{D}}\left[  B\right]
>0.287-0.265=0.022.
\]

\end{proof}

Recall that Lemma \ref{indep}\ showed the distribution $\mathcal{D}$\ to be
$2k/M$-almost $k$-wise independent. \ Examining the proof of Lemma
\ref{indep}, we can actually strengthen the conclusion to the following.

\begin{lemma}
\label{dnflem}Let $F$\ be a $k$-DNF formula, with $k\leq M/2$. \ Then%
\[
1-\frac{k}{M}\leq\frac{\Pr_{\mathcal{D}}\left[  F\right]  }{\Pr_{\mathcal{U}%
}\left[  F\right]  }\leq1+\frac{2k}{M}.
\]

\end{lemma}

\begin{proof}
Let $F=C_{1}\vee\cdots\vee C_{\ell}$. \ Fix an input $X\in\left[  M\right]
^{N}$, and suppose $F\left(  X\right)  =1$. \ Then there must be an
$i\in\left[  \ell\right]  $ such that $C_{i}\left(  X\right)  =1$. \ In the
proof of Lemma \ref{indep}, we actually showed that%
\[
\Pr_{\mathcal{D}\left(  X\right)  }\left[  C_{i}\right]  \geq1-\frac{k}{M}.
\]
It follows that%
\[
\Pr_{\mathcal{D}\left(  X\right)  }\left[  F\right]  \geq1-\frac{k}{M},
\]
and hence%
\[
\Pr_{\mathcal{D}}\left[  F\right]  \geq\Pr_{\mathcal{U}}\left[  F\right]
\cdot\left(  1-\frac{k}{M}\right)  .
\]
Similarly,%
\[
\Pr_{\mathcal{U}}\left[  F\right]  \geq\Pr_{\mathcal{D}}\left[  F\right]
\cdot\left(  1-\frac{k}{M}\right)  .
\]
The lemma now follows, using the assumption $k\leq M/2$.
\end{proof}

By combining Lemma \ref{dnflem} with the standard diagonalization tricks of
Bennett and Gill \cite{bg}, we can now prove a random oracle separation
between $\mathsf{\Pi}_{2}^{p}$\ and $\mathsf{P}^{\mathsf{NP}}$.

\begin{theorem}
\label{pnpsep}$\left(  \mathsf{\Pi}_{2}^{p}\right)  ^{A}\not \subset
\mathsf{P}^{\mathsf{NP}^{A}}$\ with probability $1$\ relative to a random
oracle $A$.
\end{theorem}

\begin{proof}
We will treat the random oracle $A$ as encoding, for each positive integer
$m$, a random sequence of integers $X_{m}\in\left[  M\right]  ^{N}$, where
$M:=2^{m}$\ and $N:=\left\lceil Mm\ln2\right\rceil $. \ Let
$f_{\operatorname*{Surj}}:\left[  M\right]  ^{N}\rightarrow\left\{
0,1\right\}  $ be our usual surjectivity function; i.e.
$f_{\operatorname*{Surj}}\left(  X_{m}\right)  =1$\ if and only if
$\operatorname{Im}_{X_{m}}=\left[  M\right]  $. \ Then let\ $L$\ be a unary
language that contains $0^{m}$\ if and only if $f_{\operatorname*{Surj}%
}\left(  X_{m}\right)  =1$. \ Clearly $L\in\left(  \mathsf{\Pi}_{2}%
^{p}\right)  ^{A}$. \ It remains to show that $L\notin\mathsf{P}%
^{\mathsf{NP}^{A}}$ with probability $1$ over $A$. \ Fix a $\mathsf{P}%
^{\mathsf{NP}^{A}}$\ machine $B^{A}$, which runs in time $p\left(  m\right)
$\ for some fixed polynomial $p$. \ Also, let $m_{1},m_{2},\ldots$\ be a
sequence of input lengths that are exponentially far apart, so that we do not
need to worry about $B^{A}\left(  0^{m_{i}}\right)  $\ querying $X_{m_{j}}%
$\ for any $j>i$. \ We will treat $X_{m_{j}}$\ as fixed for all $j<i$,\ so
that only $X:=X_{m}:=X_{m_{i}}$\ itself is a random variable. \ Then
$B^{A}\left(  0^{m}\right)  $\ makes a sequence of at most $p\left(  m\right)
$ adaptive $\mathsf{NTIME}\left(  p\left(  m\right)  \right)  $\ queries to
$X$, call them $Q_{1},\ldots,Q_{p\left(  m\right)  }$. \ For each $t\in\left[
p\left(  m\right)  \right]  $, we can write a $p\left(  m\right)  $-DNF
formula\ $F_{t}\left(  X\right)  $\ which evaluates to TRUE if and only if
$Q_{t}\left(  X\right)  $\ accepts. \ Then by Lemma \ref{dnflem}, we have%
\[
1-\frac{p\left(  m\right)  }{M}\leq\frac{\Pr_{\mathcal{D}}\left[
F_{t}\right]  }{\Pr_{\mathcal{U}}\left[  F_{t}\right]  }\leq1+\frac{2p\left(
m\right)  }{M}.
\]
This implies that%
\[
\left\vert \Pr_{\mathcal{D}}\left[  F_{t}\right]  -\Pr_{\mathcal{U}}\left[
F_{t}\right]  \right\vert \leq\frac{2p\left(  m\right)  }{M}.
\]
So by the union bound, we have%
\[
\left\vert \Pr_{\mathcal{D}}\left[  B^{A}\left(  0^{m}\right)  \right]
-\Pr_{\mathcal{U}}\left[  B^{A}\left(  0^{m}\right)  \right]  \right\vert
\leq\frac{2p\left(  m\right)  ^{2}}{M},
\]
even after we take into account the possible adaptivity of the queries.
\ Clearly $2p\left(  m\right)  ^{2}/M<0.022$\ for all sufficiently large $m$.
\ So taking the contrapositive of Lemma \ref{dvsu},%
\[
\Pr_{A}\left[  B^{A}\left(  0^{m}\right)  =f\left(  X\right)  \right]  <0.92
\]
for all sufficiently large $M$. \ So as in the standard random oracle argument
of Bennett and Gill \cite{bg}, we have%
\[
\Pr_{A}\left[  B^{A}\text{ decides }L\right]  \leq\prod_{i=1}^{\infty}\Pr
_{A}\left[  B^{A}\left(  0^{m_{i}}\right)  =f\left(  X_{m_{i}}\right)
\right]  =0.
\]
Then taking the union bound over all $\mathsf{P}^{\mathsf{NP}^{A}}$\ machines
$B^{A}$,%
\[
\Pr_{A}\left[  L\in\mathsf{P}^{\mathsf{NP}^{A}}\right]  =0
\]
as well.
\end{proof}

It is well-known that $\mathsf{P}^{\mathsf{NP}^{A}}=\mathsf{BPP}%
^{\mathsf{NP}^{A}}$\ with probability $1$ relative to a random oracle $A$.
\ Thus, Theorem \ref{pnpsep} immediately implies that $\left(  \mathsf{\Pi
}_{2}^{p}\right)  ^{A}\not \subset \mathsf{BPP}^{\mathsf{NP}^{A}}$\ relative
to a random oracle $A$ as well. \ Since the class $\mathsf{BPP}_{\mathsf{path}%
}$\ is contained in $\mathsf{BPP}^{\mathsf{NP}}$\ (as shown by Han,
Hemaspaandra, and Thierauf \cite{hht}), we also obtain the new result that
$\left(  \mathsf{\Pi}_{2}^{p}\right)  ^{A}\not \subset \mathsf{BPP}%
_{\mathsf{path}}^{A}$\ relative to a random oracle $A$.

\section{Implications for $\mathsf{AC}^{\mathsf{0}}$\label{POLY}}

In this section, we discuss two implications of our counterexample for
$\mathsf{AC}^{\mathsf{0}}$ functions.

\begin{enumerate}
\item[(1)] Linial, Mansour, and Nisan \cite{lmn}\ famously showed that every
$\mathsf{AC}^{\mathsf{0}}$\ function has average sensitivity $O\left(
\operatorname*{polylog}n\right)  $. \ By contrast, we show in Section
\ref{BS}\ that there are reasonably-balanced $\mathsf{AC}^{\mathsf{0}}%
$\ functions with average\textit{ block-sensitivity} almost linear in $n$ (on
both $0$-inputs and $1$-inputs). \ In other words, there exist $\mathsf{AC}%
^{\mathsf{0}}$\ functions that counterintuitively behave almost like the
\textsc{Parity} function in terms of block-sensitivity!

\item[(2)] Linial et al.\ \cite{lmn} also showed that every $\mathsf{AC}%
^{\mathsf{0}}$\ function can be approximated in $L_{2}$-norm by a low-degree
polynomial. \ By contrast, we show in Section \ref{LOWFAT}\ that there does
not generally exist such a polynomial that \textit{also} satisfies a
reasonable sparseness condition on the coefficients (what Aaronson
\cite{aar:ph}\ called the \textquotedblleft low-fat\textquotedblright\ condition).
\end{enumerate}

\subsection{\label{BS}The Average Block-Sensitivity of $\mathsf{AC}%
^{\mathsf{0}}$}

Let us first recall the definition of average sensitivity.

\begin{definition}
[average sensitivity]\label{sensitivity}Given a string $X\in\left\{
0,1\right\}  ^{n}$\ and coordinate $i\in\left[  n\right]  $, let $X^{i}%
$\ denote $X$\ with the $i^{th}$\ bit flipped. \ Then given a Boolean function
$f:\left\{  0,1\right\}  ^{n}\rightarrow\left\{  0,1\right\}  $, the
sensitivity of $f$\ at $X$, or $\operatorname*{s}_{X}\left(  f\right)  $, is
the number of $i$'s such that $f\left(  X^{i}\right)  \neq f\left(  X\right)
$. \ Then the average sensitivity of $f$ is%
\[
\overline{\operatorname*{s}}\left(  f\right)  :=\operatorname*{E}%
_{X\in\left\{  0,1\right\}  ^{n}}\left[  \operatorname*{s}\nolimits_{X}\left(
f\right)  \right]  .
\]
Assuming $f$ is non-constant, we can also define the average $0$-sensitivity
$\overline{\operatorname*{s}}_{0}\left(  f\right)  $\ and average
$1$-sensitivity $\overline{\operatorname*{s}}_{1}\left(  f\right)
$\ respectively, by%
\[
\overline{\operatorname*{s}}_{b}\left(  f\right)  :=\operatorname*{E}%
_{X\in\left\{  0,1\right\}  ^{n}~:~f\left(  X\right)  =b}\left[
\operatorname*{s}\nolimits_{X}\left(  f\right)  \right]  .
\]

\end{definition}

Then Linial, Mansour, and Nisan \cite{lmn}\ showed that \textit{every
}$\mathsf{AC}^{\mathsf{0}}$\textit{\ function has low average sensitivity:}

\begin{theorem}
[\cite{lmn}]\label{influence}Every Boolean function $f:\left\{  0,1\right\}
^{n}\rightarrow\left\{  0,1\right\}  $\ computed by an $\mathsf{AC}^{0}%
$\ circuit of depth $d$ satisfies $\overline{\operatorname*{s}}\left(
f\right)  =O\left(  \log^{d}n\right)  $.
\end{theorem}

We now recall the definition of \textit{block-sensitivity}, a natural
generalization of sensitivity introduced by Nisan \cite{nisan}.

\begin{definition}
[average block-sensitivity]\label{bsdef}Given a string $X\in\left\{
0,1\right\}  ^{n}$\ and a subset of indices $B\subseteq\left[  n\right]  $
(called a \textquotedblleft block\textquotedblright), let $X^{B}$\ denote
$X$\ with the bits in $B$ flipped. \ Then given a Boolean function $f:\left\{
0,1\right\}  ^{n}\rightarrow\left\{  0,1\right\}  $, the block-sensitivity of
$f$\ at $X$, or $\operatorname*{bs}\nolimits_{X}\left(  f\right)  $, is the
largest $k$ for which there exist $k$\ pairwise-disjoint blocks, $B_{1}%
,\ldots,B_{k}$, such that $f\left(  X^{B_{i}}\right)  \neq f\left(  X\right)
$ for all $i\in\left[  k\right]  $. \ Then the average block-sensitivity of
$f$ is%
\[
\overline{\operatorname*{bs}}\left(  f\right)  :=\operatorname*{E}%
_{X\in\left\{  0,1\right\}  ^{n}}\left[  \operatorname*{bs}\nolimits_{X}%
\left(  f\right)  \right]  .
\]
Assuming $f$ is non-constant, we can also define the average $0$%
-block-sensitivity $\overline{\operatorname*{bs}}_{0}\left(  f\right)  $\ and
average $1$-block-sensitivity $\overline{\operatorname*{bs}}_{1}\left(
f\right)  $\ respectively, by%
\[
\overline{\operatorname*{bs}}_{b}\left(  f\right)  :=\operatorname*{E}%
_{X\in\left\{  0,1\right\}  ^{n}~:~f\left(  X\right)  =b}\left[
\operatorname*{bs}\nolimits_{X}\left(  f\right)  \right]  .
\]

\end{definition}

We consider the following question: \textit{does any analogue of Theorem
\ref{influence} still hold if we replace sensitivity by block-sensitivity?}

We start with some simple observations. \ Call a Boolean function $f:\left\{
0,1\right\}  ^{n}\rightarrow\left\{  0,1\right\}  $
\textit{reasonably-balanced} if there exist constants $a,b\in\left(
0,1\right)  $\ such that $a\leq\operatorname*{E}_{\left\{  0,1\right\}  ^{n}%
}\left[  f\right]  \leq b$\ for every $n$. \ Then if we do not require $f$ to
be reasonably-balanced,\ it is easy to find an $f\in\mathsf{AC}^{0}$ such that
$\overline{\operatorname*{bs}}_{0}\left(  f\right)  $ and $\overline
{\operatorname*{bs}}_{1}\left(  f\right)  $\ are both large. \ For example,
the two-level AND-OR tree satisfies $\overline{\operatorname*{bs}}_{0}\left(
f\right)  =\Theta\left(  \sqrt{n}\right)  $\ and $\overline{\operatorname*{bs}%
}_{1}\left(  f\right)  =\Theta\left(  \sqrt{n}\right)  $.

So let us require $f$ to be reasonably-balanced. \ Even then, it is easy to
find an $f\in\mathsf{AC}^{0}$\ such that $\overline{\operatorname*{bs}}\left(
f\right)  =\Omega\left(  n/\log n\right)  $. \ Given an input $X=x_{1}\ldots
x_{N}\in\left[  N\right]  ^{N}$, define the \textit{Tribes function} by
$f_{\operatorname*{Tribes}}\left(  X\right)  =1$\ if there exists an
$i\in\left[  N\right]  $\ such that $x_{i}=1$, and $f_{\operatorname*{Tribes}%
}\left(  X\right)  =0$\ otherwise. \ Then not only is
$f_{\operatorname*{Tribes}}$\ in $\mathsf{AC}^{0}$, it has an $\mathsf{AC}%
^{0}$\ circuit of depth $2$ (i.e., a DNF formula). \ On the other hand, let
$X$\ be any $0$-input of $f_{\operatorname*{Tribes}}$;\ then we can change
$X$\ to a $1$-input by setting $x_{i}:=1$\ for any $i$. \ So%
\[
\operatorname*{bs}\nolimits_{X}\left(  f_{\operatorname*{Tribes}}\right)  \geq
N=\Omega\left(  \frac{n}{\log n}\right)  ,
\]
where $n:=N\log_{2}N$\ is the bit-length of $X$. \ Hence $\overline
{\operatorname*{bs}}_{0}\left(  f_{\operatorname*{Tribes}}\right)
=\Omega\left(  n/\log n\right)  $.\ \ Indeed $\overline{\operatorname*{bs}%
}\left(  f_{\operatorname*{Tribes}}\right)  =\Omega\left(  n/\log n\right)
$\ as well, since%
\[
\lim_{N\rightarrow\infty}\Pr_{X}\left[  f_{\operatorname*{Tribes}}\left(
X\right)  =0\right]  =\frac{1}{e}.
\]
By contrast, one can check that $\overline{\operatorname*{bs}}_{1}\left(
f_{\operatorname*{Tribes}}\right)  $ is only $\Theta\left(  \log n\right)  $.
\ Indeed, \textit{any} Boolean function $f$\ that can be represented by a
$k$-DNF formula satisfies $\overline{\operatorname*{bs}}_{1}\left(  f\right)
\leq k$, since if a particular $k$-term $C$ is satisfied, then there are at
most $k$ disjoint ways to make it unsatisfied.

The above observations might lead one to ask the following question:
\textit{does every reasonably-balanced }$\mathsf{AC}^{0}$\textit{\ function
}$f$\textit{\ satisfy either }$\overline{\operatorname*{bs}}_{0}\left(
f\right)  =O\left(  \operatorname*{polylog}n\right)  $\textit{\ or\ }%
$\overline{\operatorname*{bs}}_{1}\left(  f\right)  =O\left(
\operatorname*{polylog}n\right)  $\textit{? \ }We now show, alas, that the
answer is still no.

\begin{theorem}
\label{bsthm}There exists a reasonably-balanced Boolean function $f:\left\{
0,1\right\}  ^{n}\rightarrow\left\{  0,1\right\}  $, computable by a
depth-three $\mathsf{AC}^{0}$\ circuit, such that $\overline
{\operatorname*{bs}}_{0}\left(  f\right)  =\Omega\left(  n/\log n\right)  $
and $\overline{\operatorname*{bs}}_{1}\left(  f\right)  =\Omega\left(  n/\log
n\right)  $.
\end{theorem}

\begin{proof}
Let $f$ be the function $f_{\operatorname*{Surj}}$ from our counterexample.
\ As usual, we can think of an input $X$\ to $f_{\operatorname*{Surj}}$ as
belonging to either or $\left\{  0,1\right\}  ^{n}$ or $\left[  M\right]
^{N}$, where $M=2^{m}$, $N=\left\lceil Mm\ln2\right\rceil $, and $n=Nm$. \ As
in Lemma \ref{dist}, we have%
\[
\lim_{M\rightarrow\infty}\operatorname*{E}_{\left[  M\right]  ^{N}}\left[
f_{\operatorname*{Surj}}\right]  =\frac{1}{e},
\]
so $f_{\operatorname*{Surj}}$\ is reasonably-balanced.

To lower-bound $\overline{\operatorname*{bs}}_{1}\left(
f_{\operatorname*{Surj}}\right)  $, consider an input $X=x_{1}\ldots x_{N}%
\in\left[  M\right]  ^{N}$\ such that $f_{\operatorname*{Surj}}\left(
X\right)  =1$ or equivalently $\operatorname{Im}_{X}=\left[  M\right]  $.
\ Given $y\in\left[  M\right]  $,\ let $C_{y}\left(  X\right)  $\ be the set
of all $i\in\left[  N\right]  $\ such that $x_{i}=y$. \ Then we can change
$f_{\operatorname*{Surj}}\left(  X\right)  $\ from $1$ to $0$, by changing
$x_{i}$\ to an arbitrary element of $\left[  M\right]  \setminus\left\{
y\right\}  $ for each $i\in C_{y}\left(  X\right)  $. \ This implies that
$\operatorname*{bs}\nolimits_{X}\left(  f_{\operatorname*{Surj}}\right)  \geq
M$. \ Indeed, we can improve the bound to $\operatorname*{bs}\nolimits_{X}%
\left(  f_{\operatorname*{Surj}}\right)  \geq Mm$, by noticing that it
suffices to change a single \textit{bit} of $x_{i}$\ for each $i\in
C_{y}\left(  X\right)  $. \ Hence%
\[
\overline{\operatorname*{bs}}_{1}\left(  f_{\operatorname*{Surj}}\right)  \geq
Mm=\Omega\left(  \frac{n}{\log n}\right)  .
\]
Next consider an input $X=x_{1}\ldots x_{N}\in\left[  M\right]  ^{N}$\ such
that $\left\vert \operatorname{Im}_{X}\right\vert =M-1$. \ Then clearly
$f_{\operatorname*{Surj}}\left(  X\right)  =0$. Let $A\left(  X\right)  $\ be
the set of indices $i\in\left[  N\right]  $\ for which there exists at least
one $j\neq i$\ such that $x_{i}=x_{j}$. \ Then we have $\left\vert A\left(
X\right)  \right\vert \geq N-M$ by the pigeonhole principle. \ Also, for any
$i\in A\left(  X\right)  $, let $X^{i}$\ be identical to $X$, except that we
change $x_{i}$\ to the unique element of $\left[  M\right]  \setminus
\operatorname{Im}_{X}$. \ Then clearly $\operatorname{Im}_{X^{i}}=\left[
M\right]  $\ and $f_{\operatorname*{Surj}}\left(  X^{i}\right)  =1$.
\ Therefore $\operatorname*{bs}\nolimits_{X}\left(  f_{\operatorname*{Surj}%
}\right)  \geq\left\vert A\left(  X\right)  \right\vert \geq N-M$.
\ Furthermore, as in Lemma \ref{dvsu}, we have%
\[
\lim_{M\rightarrow\infty}\Pr_{\left[  M\right]  ^{N}}\left[  \left\vert
\operatorname{Im}_{X}\right\vert =M-1\right]  =\frac{1}{e}%
\]
by the Poisson approximation. \ It follows that%
\[
\lim_{M\rightarrow\infty}\overline{\operatorname*{bs}}_{0}\left(
f_{\operatorname*{Surj}}\right)  \geq\frac{1/e}{1-1/e}\left(  N-M\right)
=\Omega\left(  \frac{n}{\log n}\right)  .
\]

\end{proof}

\subsection{The Inapproximability of $\mathsf{AC}^{\mathsf{0}}$\ by Low-Fat
Polynomials\label{LOWFAT}}

Let us recall another basic result of Linial, Mansour, and Nisan \cite{lmn}.

\begin{theorem}
[\cite{lmn}]\label{lmnthm}Let $f:\left\{  0,1\right\}  ^{n}\rightarrow\left\{
0,1\right\}  $\ be computed by an $\mathsf{AC}^{0}$ circuit of depth
\thinspace$d$. \ Then for all $\varepsilon>0$, there exists a multilinear
polynomial $p:\left\{  0,1\right\}  ^{n}\rightarrow\mathbb{R}$\ of degree
$O\left(  \log^{d}\left(  n/\varepsilon\right)  \right)  $\ such that
$\operatorname*{E}_{\mathcal{U}}\left[  \left(  p-f\right)  ^{2}\right]
\leq\varepsilon$.
\end{theorem}

In this section, we ask whether one can extend Theorem \ref{lmnthm}\ to get an
approximating polynomial $p$\ that is not merely low-degree, but also
representable using coefficients that are bounded in absolute value. \ The
specific property that we want was called the \textquotedblleft
low-fat\textquotedblright\ property by Aaronson \cite{aar:ph}:

\begin{definition}
[low-fat polynomials]\label{lowfatdef}Given a multilinear polynomial
$p:\left\{  0,1\right\}  ^{n}\rightarrow\mathbb{R}$, define the fat content of
$p$, or $\operatorname*{fat}\left(  p\right)  $, to be the minimum of
$\sum_{C}\left\vert \alpha_{C}\right\vert 2^{-\left\vert C\right\vert }$\ over
all representations $p=\sum_{C}\alpha_{C}C$ of $p$ as a linear combination of
terms (that is, products of $x_{i}$'s and $\left(  1-x_{i}\right)  $'s).
\ Then we call $p$ low-fat if $\operatorname*{fat}\left(  p\right)
=n^{o\left(  1\right)  }$.
\end{definition}

One motivation for Definition \ref{lowfatdef} comes from \cite{aar:ph}, where
it was pointed out that the Generalized Linial-Nisan Conjecture is
\textit{equivalent} (via linear programming duality) to the following conjecture:

\begin{conjecture}
[Low-Fat Sandwich Conjecture]\label{lowfatconj}Let $f:\left\{  0,1\right\}
^{n}\rightarrow\left\{  0,1\right\}  $ be computed by an $\mathsf{AC}^{0}$
circuit of size $2^{n^{o\left(  1\right)  }}$. \ Then there exist low-fat
multilinear polynomials $p_{\ell},p_{u}:\left\{  0,1\right\}  ^{n}%
\rightarrow\mathbb{R}$, of degree $n^{o\left(  1\right)  }$, that
\textquotedblleft sandwich\textquotedblright\ $f$ in the following sense:

\begin{enumerate}
\item[(i)] $p_{\ell}\left(  X\right)  \leq f\left(  X\right)  \leq
p_{u}\left(  X\right)  $\ for all $X\in\left\{  0,1\right\}  ^{n}$ and

\item[(ii)] $\operatorname*{E}_{\mathcal{U}}\left[  p_{u}-p_{\ell}\right]
=o\left(  1\right)  $.
\end{enumerate}
\end{conjecture}

Without the adjective \textquotedblleft low-fat,\textquotedblright\ Conjecture
\ref{lowfatconj}\ would be equivalent to the \textit{original} Linial-Nisan
Conjecture, as shown by Bazzi \cite{bazzi}. \ And indeed, Braverman
\cite{braverman}\ heavily exploited this equivalence in his proof of the
original LN Conjecture.\footnote{Technically, Braverman constructed
polynomials that satisfied slightly different properties than (i) and (ii)
from Conecture \ref{lowfatconj}. \ However, we know from Bazzi's equivalence
theorem \cite{bazzi}\ that it must be possible to satisfy those properties as
well.}

Of course, from the fact that the GLN Conjecture is false, we can immediately
deduce that Conjecture \ref{lowfatconj} is false as well.

On the other hand, the notion of low-fat polynomials\ seems interesting even
apart from Conjecture \ref{lowfatconj}---for the low-fat condition is a kind
of \textquotedblleft sparseness\textquotedblright\ condition, which might be
useful (for example) in learning theory. \ Furthermore, the falsehood of
Conjecture \ref{lowfatconj}\ does not directly rule out the possibility of
low-fat approximating polynomials for every $\mathsf{AC}^{0}$ function, since
Conjecture \ref{lowfatconj}\ talks only about \textit{sandwiching}
polynomials. \ However, with a bit more work, we now show the existence of an
$\mathsf{AC}^{0}$\ function that has no low-fat, low-degree approximating
polynomial of any kind.

\begin{theorem}
\label{lowfatthm}There exists a Boolean function $f:\left\{  0,1\right\}
^{n}\rightarrow\left\{  0,1\right\}  $, computable by a depth-three
$\mathsf{AC}^{0}$\ circuit, for which any multilinear polynomial $p:\left\{
0,1\right\}  ^{n}\rightarrow\mathbb{R}$ that satisfies $\operatorname*{E}%
_{\mathcal{U}}\left[  \left(  p-f\right)  ^{2}\right]  =o\left(  1\right)
$\ also satisfies $\deg\left(  p\right)  \operatorname*{fat}\left(  p\right)
=\Omega\left(  n/\log^{2}n\right)  $.
\end{theorem}

\begin{proof}
Once again we let $f=f_{\operatorname*{Surj}}$. \ Let $p$ be a multilinear
polynomial such that $\operatorname*{E}_{\mathcal{U}}\left[  \left(
p-f\right)  ^{2}\right]  =\varepsilon$. \ By definition, we can write $p$ as a
linear combination of terms, $p=\sum_{C}\alpha_{C}C$, such that $\sum
_{C}\left\vert \alpha_{C}\right\vert \operatorname*{E}_{\mathcal{U}}\left[
C\right]  =\operatorname*{fat}\left(  p\right)  $. \ Hence%
\begin{align*}
\operatorname*{E}_{\mathcal{U}}\left[  p\right]  -\operatorname*{E}%
_{\mathcal{D}}\left[  p\right]   &  =\sum_{C}\alpha_{C}\left(
\operatorname*{E}_{\mathcal{U}}\left[  C\right]  -\operatorname*{E}%
_{\mathcal{D}}\left[  C\right]  \right) \\
&  \leq\sum_{C}\left\vert \alpha_{C}\right\vert \left(  \frac{2\left\vert
C\right\vert }{M}\operatorname*{E}_{\mathcal{U}}\left[  C\right]  \right) \\
&  \leq\frac{2\operatorname*{fat}\left(  p\right)  \deg\left(  p\right)  }{M},
\end{align*}
where the second line follows from Lemma \ref{indep}. \ Also, let
$\Delta:=p-f_{\operatorname*{Surj}}$. \ Then as in the proof of Lemma
\ref{dvsu}, we have%
\begin{align*}
\varepsilon &  =\operatorname*{E}_{\mathcal{U}}\left[  \Delta^{2}\right] \\
&  =\sum_{k=0}^{M}\Pr_{\mathcal{U}}\left[  \left\vert \operatorname{Im}%
_{X}\right\vert =M-k\right]  \cdot\operatorname*{E}_{\mathcal{U}}\left[
\Delta^{2}~|~\left\vert \operatorname{Im}_{X}\right\vert =M-k\right] \\
&  \geq\sum_{k=0}^{M}\frac{\operatorname*{E}_{\mathcal{U}}\left[  \Delta
^{2}~|~\left\vert \operatorname{Im}_{X}\right\vert =M-k\right]  }{e\cdot
k!}-o\left(  1\right)  ,
\end{align*}
whereas%
\begin{align*}
\operatorname*{E}_{\mathcal{D}}\left[  \Delta^{2}\right]   &  =\sum_{k=0}%
^{M}\Pr_{\mathcal{D}}\left[  \left\vert \operatorname{Im}_{X}\right\vert
=M-k\right]  \cdot\operatorname*{E}_{\mathcal{D}}\left[  \Delta^{2}%
~|~\left\vert \operatorname{Im}_{X}\right\vert =M-k\right] \\
&  \leq\sum_{k=0}^{M}\frac{\operatorname*{E}_{\mathcal{U}}\left[  \Delta
^{2}~|~\left\vert \operatorname{Im}_{X}\right\vert =M-k\right]  }%
{e\cdot\left(  k-1\right)  !}+o\left(  1\right)  .
\end{align*}
Combining, we find that%
\[
\operatorname*{E}_{\mathcal{D}}\left[  \Delta^{2}\right]  =O\left(
\varepsilon\log\frac{1}{\varepsilon}\right)  +o\left(  1\right)  .
\]
Hence%
\begin{align*}
\operatorname*{E}_{\mathcal{U}}\left[  f_{\operatorname*{Surj}}\right]
-\operatorname*{E}_{\mathcal{D}}\left[  f_{\operatorname*{Surj}}\right]   &
=\left(  \operatorname*{E}_{\mathcal{U}}\left[  p\right]  -\operatorname*{E}%
_{\mathcal{U}}\left[  \Delta\right]  \right)  -\left(  \operatorname*{E}%
_{\mathcal{D}}\left[  p\right]  -\operatorname*{E}_{\mathcal{D}}\left[
\Delta\right]  \right) \\
&  \leq\left(  \operatorname*{E}_{\mathcal{U}}\left[  p\right]
-\operatorname*{E}_{\mathcal{D}}\left[  p\right]  \right)  +\operatorname*{E}%
_{\mathcal{U}}\left[  \Delta\right]  +\operatorname*{E}_{\mathcal{D}}\left[
\Delta\right] \\
&  \leq\frac{2\operatorname*{fat}\left(  p\right)  \deg\left(  p\right)  }%
{M}+\sqrt{\operatorname*{E}_{\mathcal{U}}\left[  \Delta^{2}\right]  }%
+\sqrt{\operatorname*{E}_{\mathcal{D}}\left[  \Delta^{2}\right]  }\\
&  \leq\frac{2\operatorname*{fat}\left(  p\right)  \deg\left(  p\right)  }%
{M}+O\left(  \sqrt{\varepsilon\log\frac{1}{\varepsilon}}\right)  +o\left(
1\right)  ,
\end{align*}
where the third line follows from Cauchy-Schwarz. \ On the other hand, we know
from Lemma \ref{dist}\ that%
\[
\operatorname*{E}_{\mathcal{U}}\left[  f_{\operatorname*{Surj}}\right]
-\operatorname*{E}_{\mathcal{D}}\left[  f_{\operatorname*{Surj}}\right]
\geq\frac{1}{e}-o\left(  1\right)  .
\]
So combining, if $\varepsilon=o\left(  1\right)  $, then%
\[
\operatorname*{fat}\left(  p\right)  \deg\left(  p\right)  =\Omega\left(
\frac{M}{e}\right)  =\Omega\left(  \frac{n}{\log^{2}n}\right)  .
\]

\end{proof}

Since $f_{\operatorname*{Surj}}$\ has a depth-three $\mathsf{AC}^{0}$ circuit,
it follows from Theorem \ref{lmnthm}\ that there exists a polynomial $p$ of
degree $O\left(  \log^{3}n\right)  $\ such that $\operatorname*{E}%
_{\mathcal{U}}\left[  \left(  p-f_{\operatorname*{Surj}}\right)  ^{2}\right]
=o\left(  1\right)  $. \ Thus, one corollary of Theorem \ref{lowfatthm}\ is a
\textit{separation} between low-degree approximation and low-degree low-fat
approximation. \ In other words, there exists a Boolean function $f$\ (namely
$f_{\operatorname*{Surj}}$) that can be well-approximated in $L_{2}$-norm by a
polynomial of degree $O\left(  \operatorname*{polylog}n\right)  $, but
\textit{not} by a low-fat polynomial of degree $O\left(
\operatorname*{polylog}n\right)  $. \ This answers one of the open problems
from \cite{aar:ph}.

\section{Discussion\label{DISC}}

As we said before, we remain sanguine about the prospects for proving an
oracle separation between $\mathsf{BQP}$\ and $\mathsf{PH}$.\ \ In our view,
the lesson of our counterexample is simply that almost $k$-wise independence
is too blunt of an instrument for this problem. \ Looking at the specific
function $f_{\operatorname*{Surj}}$ in the counterexample, we find two
arguments in support of this position. \ Firstly, $f_{\operatorname*{Surj}}$
is extremely different in character from \textsc{Fourier Checking}, or any of
the other candidates for problems in $\mathsf{BQP}\setminus\mathsf{PH}$\ (such
as the ones studied by Fefferman and Umans \cite{feffumans}). \ Indeed,
$f_{\operatorname*{Surj}}$\ is not even in $\mathsf{BQP}$, as can be seen from
the BBBV lower bound \cite{bbbv} for example.\footnote{Using the BBBV lower
bound, one can show further that no $\mathsf{BQP}$\ machine can distinguish
the distributions $\mathcal{D}$\ and $\mathcal{U}$ with constant bias.}
\ Secondly, $f_{\operatorname*{Surj}}$ is trivially in $\mathsf{PH}$\ by
construction---and for that reason, our counterexample does not really say
anything unexpected about \textquotedblleft the power of $\mathsf{PH}%
$.\textquotedblright\ \ To us, the unexpected part is simply the inability of
\textit{approximate local statistics} to \textquotedblleft
certify\textquotedblright\ a problem as outside $\mathsf{PH}$, where
\textit{exact} local statistics succeed in doing so (as shown by Braverman
\cite{braverman}). \ But this is a surprise about proof techniques, not about
complexity classes.

The obvious open problems are

\begin{enumerate}
\item[(1)] to solve the relativized $\mathsf{BQP}$ versus $\mathsf{PH}%
$\ problem by whatever means, and

\item[(2)] to solve the relativized\ $\mathsf{BQP}$ versus $\mathsf{AM}%
$\ problem, \textit{possibly} by proving the depth-two GLN Conjecture.
\end{enumerate}

\noindent We reiterate our offer of a \$200 prize for problem (1)\ and a \$100
prize for problem (2).

A third interesting problem is to show that our function
$f_{\operatorname*{Surj}}\left(  X\right)  $\ cannot be computed in
$\mathsf{\Sigma}_{2}^{p}$, on a $1-\varepsilon$\ fraction of inputs
$X\in\left[  M\right]  ^{N}$. \ This would imply that $\left(  \mathsf{\Pi
}_{2}^{p}\right)  ^{A}\not \subset \left(  \mathsf{\Sigma}_{2}^{p}\right)
^{A}$\ with probability $1$\ relative to a random oracle $A$. \ A fourth
problem is whether one can say \textit{anything} nontrivial about the
block-sensitivity of $\mathsf{AC}^{0}$\ functions: for example, that every
$f\in\mathsf{AC}^{0}$\ has average block-sensitivity $O\left(  n/\log
n\right)  $.

\section{Acknowledgments}

I thank Paul Beame for sharing a draft of the manuscript \cite{beame:quantum}
with me, and Lane Hemaspaandra for pointing me to the paper \cite{cai:bh}\ of Cai.

\bibliographystyle{plain}
\bibliography{thesis}

\end{document}